\documentclass[12pt]{article}

\usepackage{eqnarray,amsmath,bm,amsfonts,amssymb}
\usepackage{float,setspace}
\usepackage{natbib}
\usepackage[normalem]{ulem}
\usepackage{multirow,array}
\usepackage[utf8]{inputenc}
\usepackage{graphicx}
\usepackage{color}
\usepackage{comment}

\usepackage{natbib}

\newtheorem{proposition}{Proposition}

\newenvironment{proof}[1][Proof]{\noindent\textbf{#1.} }{\ \rule{0.5em}{0.5em}}
\makeatletter
\renewcommand*{\@fnsymbol}[1]{\ensuremath{\ifcase#1\or *\or
 \mathsection\or \mathparagraph\or \|\or **\or \dagger\or \ddagger\or \dagger\dagger
 \or \ddagger\ddagger \else\@ctrerr\fi}}
\makeatother

\newenvironment{revs}{\begin{color}{black} \ignorespaces} 
 {\end{color}}
 
\begin{document}

\pagenumbering{Alph}
\begin{titlepage}

\title{Group Identity, Social Learning and Opinion Dynamics}

\author{Sebastiano Della Lena \thanks{Department of Economics, University of Antwerp, Prinsstraat 13, B2000 Antwerp, Belgium. Email: dellalena.sebastiano@gmail.com}
\\
Luca Paolo Merlino\thanks{Department of Economics, University of Antwerp, Prinsstraat 13, B2000 Antwerp, Belgium. Email: LucaPaolo.Merlino@uantwerpen.be}}
\maketitle

\abstract{In this paper, we study opinion dynamics in a balanced social structure consisting of two groups. Agents learn the true state of the world naively learning from their neighbors and from an unbiased source of information. Agents want to agree with others of the same group---\textit{in-group identity},--- but to disagree with those of the opposite group---\textit{out-group conflict}. We characterize the \begin{revs}long-run\end{revs} opinions, and show that agents' influence depends on their Bonacich centrality in the signed network of opinion exchange. Finally, we study the effect of group size, the weight given to unbiased information and homophily when agents in the same group are homogeneous.
	}
\thispagestyle{empty}

\end{titlepage}

\pagenumbering{arabic}

\section{Introduction}\label{introduction}

\begin{revs}
Society is divided in socio-economic or demographic groups who potentially hold different opinions on some specific issues. 
The interaction between these groups are often characterized by conflict spurred by these different views. Indeed, groups with different cultures of systems of belief can perceive potential threats when interacting with members of other groups \citep{stephan2017}. 
%
These threats rationalize the desire to agree with members of their own group, and to disagree with members of the other group.
%
Hence, group identity translates into a desire to adhere to a social norm, that is, to hold a similar opinion to the other members of the group, while being willing to systematically and actively differentiate from the members of the other groups with a different set of values \citep{akerlof2000}.\footnote{ \cite{dellalena2021} provides a microfoundation of the acculturation processes that lead two groups to be identify with similar or different norms, depending on the strategic environment where agents interact.}

For example, information often diffuses online in echo chambers whose members have similar views. Surprisingly, in communities supporting unscientific claims, interaction with debunking post seems to lead to an increasing interest in conspiracy-like content \citep{zollo2017}. In other words, interaction with  members of the same group and of the other groups has different effects on opinions, and it results in dissenting information being often ignored.


The different nature of these interactions can be represented as signed links. A positive link is then associated with positive sentiments---such as conformism, willingness to compromise and cooperation,---while a negative link is associated with negative sentiments---such as antagonism, coercion, or conflict \citep{cartwright1956}.

\end{revs}


Signed networks help to understand group dynamics and conflict. For example, in the famous study of \cite{sampson1968}, monks were asked the three people they liked and disliked the most among the other monks at their monastery. Their responses showed the existence of a balanced sociometric structure with two cliques. Furthermore, \cite{bonacich2004} showed that the status ranking of monks according to a modified network centrality allowing for negative links correctly predicts the four monks that would be expelled later on by giving them the lowest status.

However, to the best of our knowledge, these findings have not yet been understood in a formal model of opinion dynamics. This is precisely the contribution of this paper.

In the model, we study opinion dynamics in a structurally balanced society divided in two groups.\footnote{A signed network is structurally balanced if all positive relations connect members of the same clique and all negative relations connect members of different cliques. This property tends to be satisfied in signed networks \citep{harary1953}.} Individuals aim to learn the unknown state of the world combining the opinions of others in the society and the information they receive from an unbiased source of information, which can be thought of as an unbiased media outlet. The novelty of this paper is that individuals would like to agree with members of their own group---which we name the \textit{in-group identity effect},---and to disagree with those of the other group---which we name the \textit{out-group conflict effect}.

Starting from the socialization network that represents who interacts with whom, we superpose these inter-group conflict effects. By doing so, we derive the interaction networks that is relevant for social learning, which we name the \textit{opinion exchange network}.

The opinion exchange network is then a signed and structurally balanced network. This results because all positive relations connect members of the same group, while all negative relations connect members of different groups. Structural balance is a property that holds in a series of applications, ranging from international relationships, starting from the seminal contribution of \cite{harary1961}, to the aforementioned interactions in monasteries \citep{sampson1968}, and online social networks \citep{guha2004}, and it has been applied to many other contexts since.\footnote{See \cite{easley2010} for an introduction to the literature of structural balance.} 

The main result of this paper is that opinions in the society converge in \begin{revs}the long-run\end{revs} to a vector that is proportional to the \begin{revs}weighted\end{revs} Bonacich centrality of individuals in the opinion exchange network, which allows for positive and negative links. Interestingly, the opinion leaders in this signed network\begin{revs}, i.e., those who better aggregate information,\end{revs} are potentially different than those who would be predicted in models where the sign of links (or negative links altogether) are not considered.

While in-group identity helps individuals to learn the truth, inter-group conflict makes everyone worse off, in the sense that opinions \begin{revs}in the long-run\end{revs} are farther away from the true state of the world. Additionally, individuals who engage in more conflict are those farthest away from the truth, and reducing conflict enhances social learning in a group.

We also study the impact of the group size, the weight given to unbiased information and homophily assuming that the network of interaction is symmetric among players of the same group. We find that individuals' ability to learn the truth increases in their group size, as this strengthen the identity effect while dampening the inter-group conflict effect. A similar result holds for the weight agents in a group give to unbiased information. Interestingly, increasing such a weight reduces the ability to learn of the other group. Indeed, as a group gets better information, the other becomes less central in the opinion exchange network, thereby getting pushed away from the truth. Finally, homophily has a positive effect on social learning in a group, and negative in the other, as interacting less with the other group reduces the salience of inter-group conflict.

\begin{revs}
Lastly, we discuss some extensions of the baseline model. In particular, we characterize long-run opinions when agents belonging to one of the two groups have a biased private source of information. We show that such bias may induce group polarization and push the unbiased group further from the truth than the biased one.
\end{revs}

This paper contributes to the literature of non-strategic opinion dynamics, initiated by \cite{degroot1974}, and recently revived by \cite{golub2010} and \cite{molavi2018}, among others. In particular, as in \cite{sandroni2012} and \cite{dellalena2019}, we model social learning as a convex combination of average-based updating of neighbors' opinions \citep[e.g.,][]{degroot1974,golub2010} and an unbiased source of information.\footnote{Differently from \cite{sandroni2012} and \cite{dellalena2019}, we consider a non-probabilistic setting with fully precise signals from the unbiased source of information. This simplifies individuals' Bayesian updating rule.} \begin{revs}In \cite{forster2016}, agents have a De Groot learning rule, but they are strategic in manipulating the trust in the opinions of others in their favor, and show that manipulation can connect a segregated society and thus lead to mutual consensus.\end{revs} See \cite{grabisch2020} for a recent survey of the literature on nonstrategic models of opinion dynamics. In this literature though, only positive relationships/links are usually allowed for.\footnote{\begin{revs} Because $complementarity$  in the opinions of agents belonging to different groups is captured by standard model of opinion dynamics, where the network is described by non-negative matrices \citep[e.g.,][]{degroot1974, sandroni2012}, in this paper, we focus only the case in which there is $substitution$ in the opinion of agents belonging to different groups.
\end{revs}}

Recently however, several researchers studied the effect of negative relationships, but only to study anti-conformity, that is, the effect on opinion dynamics of individuals who like to disagree with others. For example, \cite{buechel2015} study a model in which agents update opinions by averaging over their neighbors’ expressed opinions, but may misrepresent their own opinion by conforming or counter-conforming with their neighbors. In \cite{grabisch2019}, there are agents who are either conformist or anti-conformist (or a mixture of the two). In this paper, agents' desire to conform or not to an opinion depends on their own identity, and on the identity of the agent who they hear the opinion from. So, in a sense everyone is both a conformist and an anti-conformist depending on with whom they interact. In a similar spirit to our paper, \cite{shi2019} review the convergence properties of opinion dynamics in a signed network; however, the possibility to have an unbiased source of information is not yet explored.

\begin{revs}
The literature of network formation mostly addressed the creation of positive links.\footnote{See \cite{jackson} for a review of the literature, \cite{herings2009} for models with farsighted players, and \cite{KM,kinateder2022} for models in which cooperation is modeled as a local public good.} In \cite{demarti2017}, study a model where the cost of a link between two agents from different groups diminishes with the rate of exposure to the other group. They show that the two groups can be integrated or segregated depending on the inter-group costs, and that links that connect both communities prevail.\footnote{\begin{revs}
\cite{foerster2021} introduce the possibility to establish links not observed by others in this framework. They show that there are pairwise stable equilibria in which players who benefit less from connecting to members of the other group are segregated.\end{revs}} More related to our paper, negative relationships has been taken into account formally by \cite{hiller2017}, who provided a microfoundation of structural balance \citep{cartwright1956}. In this paper, motivated by the work of \cite{sampson1968} and \cite{bonacich2004}, we take the existence of balanced structure as given, and focus instead on their effects on opinion dynamics.
\end{revs}

This paper also relates to the literature about the spread of misinformation in networks. As here, it is the existence of groups that leads to the existence of different opinions in the economy \citep{bloch2018,tabasso2020,dellalena2019}. However, while there the groups differ in their bias with respect to the state of the world, here groups differ in their identity, which leads to inter-group conflict.

The paper proceeds as follows. Section \ref{model} introduces our model. Section \ref{main} presents the main results of the general model. Section \ref{special} highlights some features of the model using some specific patterns of interactions. \begin{revs}Section \ref{discussion} provides some discussion of the main assumptions of our model.\end{revs} Section \ref{conclusions} concludes.

\section{The Model}\label{model}

Let us denote by $\Theta\subset \mathbb{R}$ the set of states of the world, and by $\theta^*$ the true state of the world, which is unknown to the agents.

There is a set $N = \{1,2,...,n\}$ of agents who are partitioned in two groups $\mathcal{C}:=\{A,B\}$, of size $n_A$ and $n_B$ respectively, so that $n_A+n_B=1$. We order agents in such a way that $A=\{1,...,n_A\}$ and $B=\{n_A+1,...,n\}$.

Agents want to learn the true state of the world, which we denote by $\theta^*$. They do so aggregating information from two sources: either exchanging their opinion with their social contacts, with whom they may want to agree or disagree, or by directly observing an unbiased source of information.

Regarding the first source, agents interact in a network of social interactions. The interaction patterns are captured by an $n \times n$ non-negative matrix $\bm{W}$, where $w_{ij}\in[0,1]$ indicates to which extent $i$ pays attention to $j$'s opinion, for each $i,j\in N$. The interaction matrix $\bm{W}$ can be written as follows:
\begin{eqnarray}\label{w}
\bm{W} :=\begin{bmatrix}
\bm{W}^{AA} &\bm{W}^{AB} 
\\
\bm{W}^{BA} & \bm{W}^{BB}\end{bmatrix},
\end{eqnarray}
so that $\bm{W}^{AB}$ captures the interactions of an agent $i\in A$ with an agent of group $j\in B$---the other sub-matrices of $\bm{W}$ are interpreted in the same way. Note that $\bm{W}$ can be asymmetric.

Additionally, agents have access to an unbiased source of information that reveals that the true state of the world $\theta^*$. We denote by $w_i$ the weighing that $i$ gives to $\theta^*$, i.e., how much $i$ is exposed to the truth or pays attention to the unbiased source of information. So, $w_i$ can be interpreted also as the probability that $i$ assigns to the unbiased source of information to be effectively unbiased. Let $\mathbf{w}_C:=(w_i)_{i \in C}$ for $C=A,B$ be the vector describing the exposure that each agent $i\in C$ has to the true state of the world. Similarly, $\mathbf{w}:=(\mathbf{w}_A,\mathbf{w}_B)$ is the vector containing the exposure of all agents.

\begin{revs}
We assume that for each $i \in N$ $\sum_{j \in N / \{i\}} w_{ij}+w_i=1$ so that the matrix $\bm{W}$ is sub-stochastic. \end{revs}

Each agent $i\in N$ has an opinion, which we denote by $\mu_{i,t}$. Let $\bm{\mu}_t:= (\bm{\mu}_{A,t},\bm{\mu}_{B,t})$ be the column vector of opinions of agents $A$ and $B$ at time $t$, where $\bm{\mu}_{C,t}:=(\mu_{i,t})_{i \in C}$, for $C=A,B$.

Agents update their opinion by weighting the opinions they hear from members of their own group, opinions they hear from members of the other group, and of the unbiased source of information.

Inspired by \cite{stephan2017} and \cite{akerlof2000}, we assume agents identify themselves with their group, and this dictates their behavior regarding with whom they have to agree. So, in our model of opinion dynamics, identity affects individual opinions.

More in detail, we assume that, as in standard models of opinion updating \citep{golub2010,sandroni2012}, agents desire to agree with members of their group regarding the true state of the world. However, they have a desire not to conform with the opinions of members of the other group. Additionally, agents of different group might believe to a different extent to the unbiased source of information, for example, because they have different sentiments regarding how much it represents their group identity.

More formally, the opinions of agents of groups $A$ and $B$, $\bm{\mu}_{A}$ and $\bm{\mu}_{B}$ respectively, evolve according to the following equations:
\begin{align}
\label{eq:updatingA}
     \bm{\mu}_{A,t}= & \alpha \bm{W}^{AA} \bm{\mu}_{A,t-1} + \beta \bm{W}^{AB}\bm{\mu}_{B,t-1} + (1- \alpha - \beta)\mathbf{w}_A \theta^*,
      \\
          \bm{\mu}_{B,t}=& \alpha \bm{W}^{BB}\bm{\mu}_{B,t-1} + \beta \bm{W}^{BA} \bm{\mu}_{A,t-1} + (1- \alpha - \beta) \mathbf{w}_B \theta^*,
      \label{eq:updatingB}
\end{align}
where $ \alpha \in [0,1], \beta \in [-1,0] $ are the intensity of \textit{in-group identity} and \textit{out-group conflict}, respectively, and they are such that $(1-\alpha-\beta)\in[0,1]$.\footnote{Notably, how each agent uses the opinions of others belonging to other group when updating her opinion is consistent with the repelling rule in \cite{shi2019}. \begin{revs}
The main difference is that here each agent has also a certain level of exposure to a true source of information.\end{revs}}

Given this process of opinion dynamics, we can derive from the matrix of interaction $\bm{W}$ the matrix of exchange of opinions $\bm{\tilde{W}}$ representing how one agent's opinion affects another depending on their respective group identity. Hence,
\begin{eqnarray}\label{tildew}
\bm{\tilde{W}} :=
\begin{bmatrix}
\alpha \bm{W}^{AA} & \beta \bm{W}^{AB} 
\\
\beta \bm{W}^{BA} & \alpha \bm{W}^{BB}.
\end{bmatrix}
\end{eqnarray}
Each agent has positive (or zero) links with members of the same group, and negative (or zero) links with those of the other group.\footnote{\begin{revs}See \cite{BlochDutta} and \cite{kinateder2022} for a microfoundation of weighted links in the transmission of information.\end{revs}} Hence, by Theorem 3 in \cite{harary1953}, $\bm{\tilde{W}}$ represents a \textit{structurally balanced network}, i.e., a network in which the product of the sign of the edges in a cycle is positive for all possible cycles in $\bm{\tilde{W}}$.\footnote{For a formal definition, let us consider a subset of distinct nodes $S:=\{1,2,..,K-1,K\}$. There is \textit{path} among nodes in $S$ if for each $k \in S$,$w_{k-1 k} w_{k k+1} \neq 0$. If $\prod_{k =2}^{K-1} w_{k-1 k} w_{k k+1} > (<)0$ the \textit{sign of the path} is positive (negative). A \textit{cycle} is a path that begins and ends at the same vertex. Thus, the \textit{sign of the cycle} is the sign of the associate path.}

Similarly, we write the vector containing the weight that agents of group $A$ and $B$ give to the true state of the world as follows:
\begin{eqnarray}\label{unbiased}
\tilde{\mathbf{w}}:=(1- \alpha- \beta)\mathbf{w}.
\end{eqnarray}

Our model of identity with out-group conflict then gives the following equation of opinion dynamics
\begin{equation}\label{eq:dyn}
    \bm{\mu}_{t}= \bm{\tilde{W}} \bm{\mu}_{t-1} + \tilde{\mathbf{w}} \theta^*.
\end{equation}
The first term in \eqref{eq:dyn} describes the role if interpersonal exchanges of opinions in determining one's belief, while the second term describes the importance of the unbiased source of information.

We are interested in \textit{\begin{revs}the long-run\end{revs} opinions}, which we denote by $\bm{\mu}$. In other words, if the opinions in the economy at time $t$ are described by $\bm{\mu}$, then they will also be $\bm{\mu}$ at $t+1$.

\subsection{Interpretation}\label{interpretation}

The process of opinion updating described by equations \eqref{eq:updatingA} and \eqref{eq:updatingB} can be interpreted as the result of the \textit{myopic best-response dynamics} in a \textit{in-group} coordination and \textit{out-group} anti-coordination game. Let us consider, for each generic $i \in C$, the following utility function
\begin{eqnarray}
u_{i}\big(\bm{\mu}; \theta^*\big)&=& -
\underbrace{ \alpha \sum_{j \in C} w_{ij}\big(\mu_{i}- \mu_{j}\big)^2}_{\substack{\text{in-group \ identity \ ($\alpha >0$)}}} - \underbrace{ \beta \sum_{z \in \mathcal{C} \setminus C} w_{iz}\big(\mu_{i}- \mu_{z}\big)^2}_{\substack{\text{out-group \ conflict \ ($\beta<0$)}}} \notag \\ &-&\underbrace{(1-\alpha-\beta) w_i \big(\mu_{i}- \theta^*\big)^2}_{\substack{\text{distance \ from \ the \ truth \ $(1-\alpha-\beta>0$)}}}.
\label{eq:utility}
 \end{eqnarray}
As in a standard coordination game, agents receive disutility from having opinions different from those of the agents belonging to the same group and from being distant from the true state of the world. Moreover, out-group conflict makes agents gain utility from being distant from the opinion of members of the other group. The weighting of the three component of the utility function depends from the physical network of the interactions, $\bm{W}$, the exposition to the unbiased source of information $w_i$ and the intensity of in-group identity and out-group conflict, $\alpha, \beta$.

Given the utility \eqref{eq:utility}, the best reply of each $i\in N$ is 
\begin{equation}
\mu_{i}= \sum_{j \in C} \tilde{w}_{ij} \mu_{j} + \sum_{z \in \mathcal{C} \setminus C} \tilde{w}_{iz} \mu_{z} + \tilde{w}_i\theta^*.
\label{eq:br}
 \end{equation}
If agents myopically best reply to others' opinions, aggregating for each agent in groups $A$ and $B$, we get equations \eqref{eq:updatingA} and \eqref{eq:updatingB} and, thus, \eqref{eq:dyn} of the opinion dynamics in our model.

\section{Results}\label{main}

\begin{revs}
Before presenting the results of the opinion dynamics \eqref{eq:dyn} at steady state, it is convenient to define the matrices $\bm{B}:=\left( \bm{I} -\bm{W}\right)^{-1}$ and $\tilde{\bm{B}}:=\left( \bm{I} -\bm{W}\right)^{-1}$ with generic elements $b_{ij}$ and $\tilde{b}_{ij}$, respectively. Moreover, we denote as $\mathbf{b}:=\bm{B}\cdot \mathbf{1}$ the vector containing the Bonacich centralities of agents in the social interaction network $\bm{W}$ and as $\tilde{\mathbf{b}}: = \tilde{\bm{B}} \cdot \tilde{\mathbf{w}}$ the vector containing the $weighted$ Bonacich centralities of agents in the opinion exchange network $\tilde{W}$. We sometimes refer to $\tilde{\mathbf{b}}$ as the \textit{modified (weighted) Bonacich centrality}.
\end{revs}

\begin{proposition}
\label{prop:1}
The opinions dynamics defined in equation \eqref{eq:dyn} always converge and, at the steady state, we have 
\begin{equation}
 \begin{revs} \Rightarrow \bm{\mu}= \tilde{\mathbf{b}}\theta^*.\end{revs}
 \label{eq:ss}
\end{equation}
Moreover, equation \eqref{eq:ss} is also the unique Nash Equilibrium of the game in equation \eqref{eq:utility} for all $i \in N$.
\end{proposition}

\begin{proof}
\begin{revs}
Let us iterate equation \eqref{eq:dyn}:

\begin{align*}
    \bm{\mu}_{t+1}=& \bm{\tilde{W}} \bm{\mu}_{t} + \tilde{\mathbf{w}} \theta^*,
    \\
    \bm{\mu}_{t+2}=& \bm{\tilde{W}}\big( \bm{\tilde{W}} \bm{\mu}_{t} + \tilde{\mathbf{w}} \theta^*\big) + \tilde{\mathbf{w}} \theta^*,
     \\
     ...&
     \\
      \bm{\mu}_{t+T}=&\bm{\tilde{W}}^{T}\bm{\mu}_{t} + \sum_{\tau=0}^{T-1}\bm{\tilde{W}}^{\tau} \tilde{\mathbf{w}} \theta^*,
\end{align*}
so that
\begin{align*}
  \lim_{T \rightarrow \infty} \left(
      \bm{\mu}_{t+T}\right)=&  \lim_{T \rightarrow \infty} \left(\bm{\tilde{W}}^{T}\bm{\mu}_{t} + \sum_{\tau=0}^{T-1}\bm{\tilde{W}}^{\tau} \tilde{\mathbf{w}} \theta^*\right)
      \\
      =& \underbrace{\left( \bm{I}-\bm{\tilde{W}}\right)^{-1}\tilde{\mathbf{w}}}_{ \tilde{\mathbf{b}}} \theta^* .
\end{align*}
Note that, given the assumption that, for each $i \in N$, $\sum_{j \in N / \{i\}} w_{ij}+w_i=1$, the matrix $\bm{\tilde{W}}^{T}$ is sub-stochastic, thus, its spectral radius is less then 1. These facts imply that $\lim_{T \rightarrow \infty} \bm{\tilde{W}}^{T}\bm{\mu}_{t} =\bm{0}$.
\\
Similarly, as for the convergence of geometric series of matrices, given the spectral radius of $\bm{\tilde{W}}^{T}$, we get $\lim_{T \rightarrow \infty}  \sum_{\tau=0}^{T-1}\bm{\tilde{W}}^{\tau}=\left( \bm{I}-\bm{\tilde{W}}\right)^{-1}$. 
\\
Once this result is established, it is enough to note that $\left( \bm{I} -\bm{\tilde{W}}\right)$ is always invertible to establish the result.
\end{revs}
\end{proof}

Thus, the \begin{revs}the long-run\end{revs} opinion vector depends on the modified Bonacich centrality of each agent and on the true state of the world. \begin{revs}The proof establishes that $\bm{\tilde{W}}$ is a a sub-stochastic matrix. Hence, the joint best response function is a contracting mapping and, by Banach fixed-point theorem, there is a unique Nash Equilibrium.\end{revs}

Equation \eqref{eq:ss} shows that, if agents have the same level of exposure to $\theta^*$, the one with higher modified Bonacich centrality will be the one with the opinion closer to the true state of the world. Interestingly, the agents with higher modified Bonacich centrality, who are thus the opinion leaders---i.e., the ones who better aggregate social information and are closer to the truth,---are different from the most central agents in the network of social interactions. 
Hence, introducing negative relationship among agents in the network has a huge impact on the opinion dynamics and the centrality of agents.\footnote{\begin{revs} This result relates to \cite{shi2019}. The difference is that in the present work we postulate that $w_i>0$ for at least one $i \in N$, thus the long-run opinion vector does not depends on the vector of the initial opinions, but on the true state of the world and on how agents are directly or indirectly exposed to it. Moreover, since the interaction matrix is sub-stochastic, agents' opinions depends on a Bonacich centrality and not on the eigenvector centrality. Indeed we can easily understand from equation (15) in \cite{bonacich2001} that the vector $\Big(\bm{I} -\bm{\tilde{W}}\Big)^{-1} \cdot \bm{1}$ corresponds to the vector of eigenvector centralities if and only if the largest eigenvector of $\bm{\tilde{W}}$ is equal 1. See also the discussion on these different measures of centrality provided in \cite{dasaratha2020}.\end{revs} }
 
\begin{revs}
The intuition is that agents repel opinions from the other group independent of the opinion itself. In other words, even if the agents of the two groups would have the same opinion, then those of the opposing group are incorporated with negative weight. This leads to a distortion of the true state of the world in everyone's opinions. 

In line with this reasoning, as\end{revs} $\tilde{\mathbf{b}}=\sum_{t=0}^\infty\bm{\tilde{W}}^t \tilde{\mathbf{w}}$, the farther an agent is in the network from the agents belonging to the other group, the highest her modified Bonacich centrality will be and the closer she will be to $\theta^*$. These results stem from the fact that direct interactions with agents of the other group trigger feelings of antagonism and anger inducing negative reactions that push one away from the true state of the world. The negative effects of inter-group conflict are mitigated the further an agent is from the other group.

We now provide \begin{revs}two\end{revs} examples that show how different network structures together with agents' identities may impact on agents' modified Bonacich centrality and, thus, on \begin{revs}the long-run\end{revs} opinions. In the following examples, we assume for simplicity that $\alpha=-\beta=1$, so that $(1- \alpha- \beta)=1$ and, thus, 
\[
\bm{\tilde{W}} =
\begin{bmatrix}
 \bm{W}^{AA} & - \bm{W}^{AB} \\
-\bm{W}^{BA} & \bm{W}^{BB} \end{bmatrix},
\]
and $\tilde{\mathbf{w}}=\mathbf{w}$. We also assume that, for all the agents, each link to other agents $w_{ij}$, the self-loop $w_{i}$, and the weighting to the unbiased source of information $w_{i}$ is the same and equal to
\[
\frac{1}{\# \ of \ links \ + \ 2}.
\]
\\
\noindent \textbf{Example 1 (Complete Network)} We now develop an example of a society composed by $4$ agents in which the matrix of interaction $\bm{W}$ is complete, meaning that $w_{ij}=1/5$ for all $i,j\in N$. The resulting matrix of opinion exchange $\bm{\tilde{W}}$ is then a complete signed network. In particular, in Figure \ref{fig:complete} we represent such a network. Agent $1$, in red, is the only member of groups $A$, while agents $2$, $3$ and $4$ belong to group $B$ (depicted in blue). Table \ref{tab:complete} shows the Bonacich centralities of agents in the interaction network $\mathbf{b}$ 
and \begin{revs}the long-run opinion vector $\bm{\mu}$, which depends on the agents' weighted modified Bonacich centralities in the network of opinion exchange, $\tilde{\mathbf{b}}$ \end{revs}. We can see that the minority (agent $1$) becomes radicalized, and is very far from the truth. The interpretation is that agent $1$ is strongly affected by the conflict stemming from the contact with agents in the other group. Agents belonging to the majority (blue), having relatively less contact with agents of the other group, are closer to the truth.

It is also interesting to note that all agents would have the same Bonacich centrality if we were to ignore the sign of the links, that is, the tension between the identity and the inter-group conflict effects. On the contrary, taking these into account shows that the opinion leaders in the network depicted in Figure \ref{fig:complete} belong to the majority. In Section \ref{special}, we will show more general results on how group size affects opinion dynamics.
\\
\\
\begin{figure}[!ht]
\centering
\includegraphics[scale=0.3]{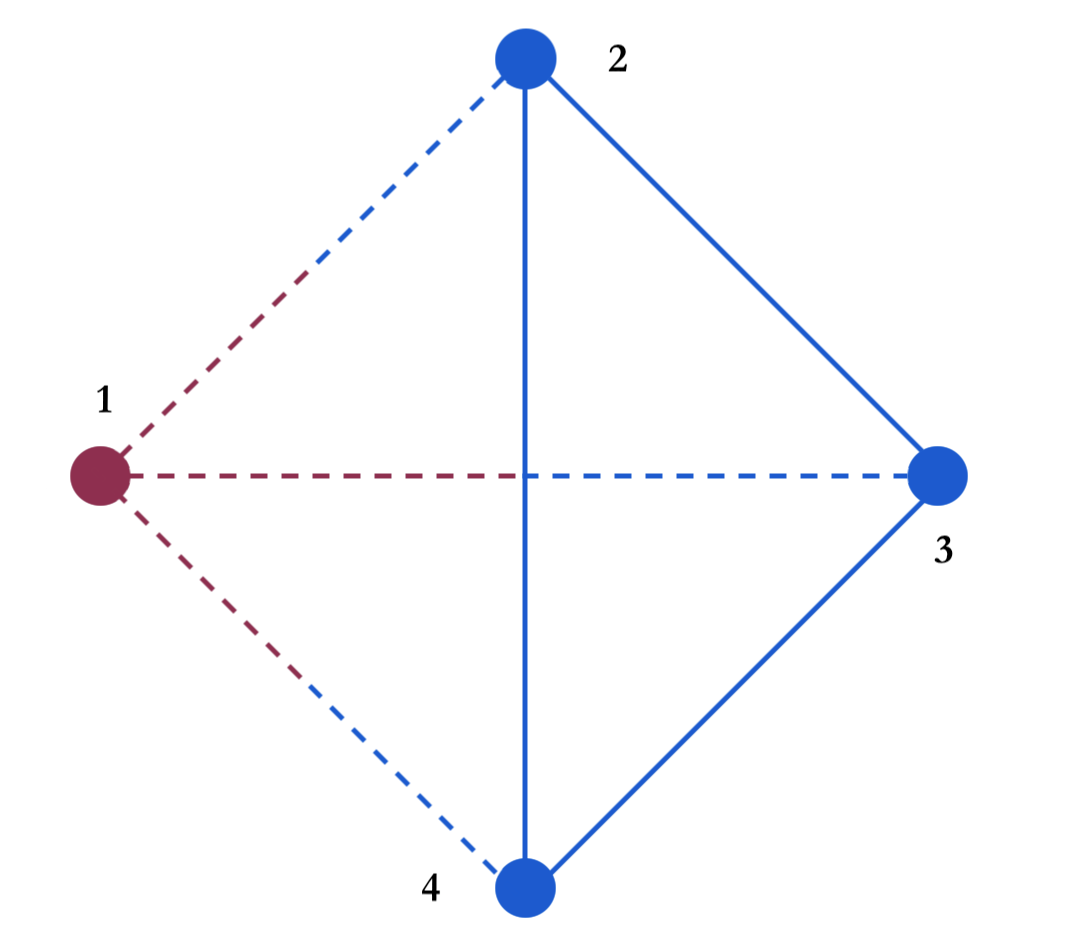}
\caption{A signed complete network of $4$ agents. Dotted edges represent negative links, full edges represent positive links. Nodes in red belong to group $A$, those in blue to group $B$.}
\label{fig:complete}
\end{figure}
\begin{table}[!ht]
\begin{center}
\begin{tabular}{|c|c|c|}
\hline \rule{0pt}{10pt}
\emph{Nodes} & Bonacich centrality, $\mathbf{b}$ & \begin{revs}Long-run\end{revs} opinions, $\bm{\mu}$ \\ \hline \hline
1 & $5$ & $-1/5$ 
\\ \hline
2 & $5$ & $3/5$ 
\\ \hline
3 & $5$ & $3/5$ 
\\\hline
4 & $5$ & $3/5$ 
\\ \hline
\end{tabular}%
\end{center}
\bigskip
\caption{Bonacich centralities and long-run opinions in the complete symmetric network depicted in Figure \ref{fig:complete} with $\theta^*=1$, $\alpha=-\beta=1$ and $w_C=1/5$ for $C=A,B$.}
\label{tab:complete}
\end{table}

\noindent \textbf{Example 2 (Ring Network)} Let us now consider a society composed by $6$ agents connected in a circle as shown in Figure \ref{fig:ring}. Agents $1$, $2$, and $3$ belong to group $A$ (depicted in red), whereas agents $4$, $5$, and $6$ belong to group $B$ (depicted in blue). From Table \ref{tab:ring} we can see that the Bonacich centralities in the social interaction network is the same for all the agents, but, considering the opinion exchange network, agents $2$ and $5$ are the opinion leaders as they are the most central agents. The intuition is that people with no direct interaction with the opposite group are the less affected by out-group conflict and, thus, the most capable of aggregating information. As a result, their opinion is the closest to the truth. 

\begin{figure}[ht!]
\centering
\includegraphics[scale=0.3]{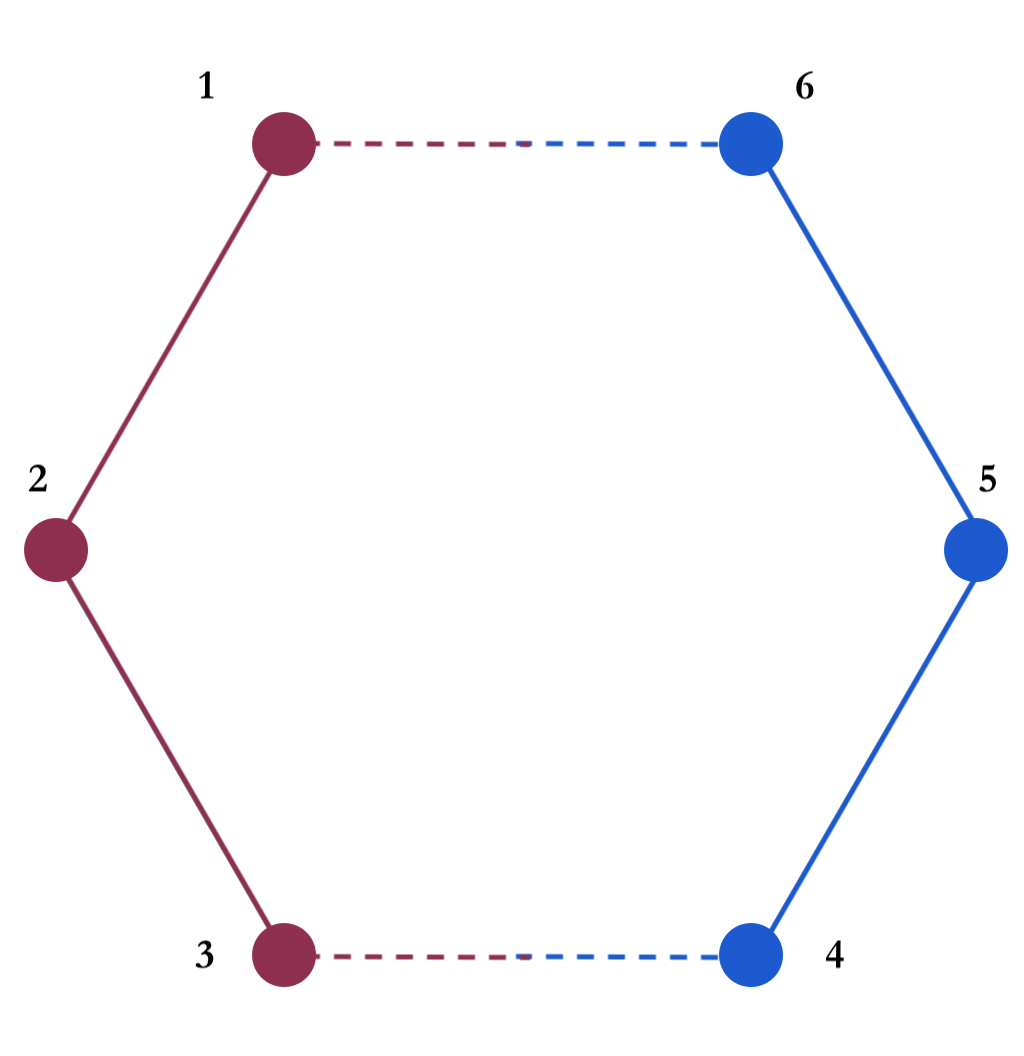}
\caption{A signed ring network of $6$ agents. Dotted edges represent negative links, full edges represent positive links. Nodes in red belong to group $A$, those in blue to group $B$.}
\label{fig:ring}
\end{figure}
\begin{table}[ht!]
\begin{center}
\begin{tabular}{|c|c|c|}
\hline \rule{0pt}{10pt}
\emph{Nodes} & Bonacich centrality, $\mathbf{b}$ & \begin{revs}Long-run\end{revs} opinions, $\bm{\mu}$ \\ \hline \hline
1 & 4 & $2/5$ 
\\ \hline
2 & 4 & $3/5$ 
\\ \hline
3 & 4 & $2/5$ 
\\\hline
4 & 4 & $2/5$ 
\\ \hline
5 & 4 & $ 3/5$ 
\\ \hline
6 & 4 & $2/5$ 
\\ \hline
\end{tabular}
\end{center}

\bigskip
\caption{Bonacich centralities and long-run opinions in the ring network depicted in Figure \ref{fig:ring} with $\theta^*=1$, $\alpha=-\beta=1$ and $w_C=1/4$ for $C=A,B$.}
\label{tab:ring}
\end{table}

\section{An Application to Homogeneous Groups}\label{special}

In this section, we explore some specific network structures that have the property that the patterns of interaction are the same for all agents of the same group. The objective is twofold. First, we develop some simple opinion dynamics that we believe could be useful for applications of our model. Second, these simplified frameworks allow us to derive additional comparative statics results, and, in particular, to investigate the role of group size, the weight agents assign to the public signal and homophily, i.e., the probability of interacting with someone of the same vs. the other group, on opinion dynamics.

The idea is that people meet each other randomly, and that people of the same group have a similar meeting technology. For example, people of similar socio-economic status attend similar schools or join the same clubs, so that they end up having comparable social circles. Note also that such structures would emerge in symmetric equilibria of network formation games in which players decide a generic socialization investment that determines how many random meetings they have.\footnote{Some examples include \cite{galeotti2014} and \cite{merlino2014,merlino2019}, or \cite{bolletta2021} for an application with two groups.}

Additionally, we assume that agents belonging to the same groups have the same exposition to an unbiased source of information (e.g., same access to media, cultural leaders that share information in the communities, same propensity towards the truth, etc.).

More formally, we assume that a \textit{society with two homogeneous groups} is one where $\alpha=-\beta=1$ and $\bm{W}$ is such that $w_i=w_A$ for all $i \in A$ and $w_i=w_B$ for all $i \in B$, and $w_{ij} = w_{zj}$ for each $i, z \in C$ and $j\in N \setminus \{i\}$, with $C=A,B$.

Under these assumptions, each agent in a group will have the same opinions. In this section, we denote by $\mu_C$ \begin{revs}the long-run\end{revs} opinion of each $i\in C$, for $C=A,B$; in other words, $\mu_C$ is the opinion in \begin{revs}long-run\end{revs} of the representative agent of group $C$.

In this case, at each $t$, the opinions of each agent are affected by the average opinion of the two groups at the previous time, thus they became homogeneous within each group. This simplification then allows us to study the opinion dynamics of a representative agent for each group.

To keep track of a group's average opinion, we define $\bar{\mu}_{C,t-1}:= \sum_{i\in C } \mu_{i,t-1}/n_C$ and $\bar{\bm{\mu}}_{C,t-1}=\bar{\mu}_{C,t-1}\bm{1}$, for $C=A,B$.

\subsection{Group size and Weight of the Public Signal}

In the first example, we focus on the role played on opinion dynamics by group size and the weight agents of a group assign on the public signal. For simplicity, we additionally assume that agents meet as frequently agents of the same and of the opposite group. More formally, we say there is \textit{no homophily} if \begin{eqnarray*}
w_{ij} &=& \frac{1-w_A}{n} \text{ for each } i, j \in A \text{ and } j\in N \setminus \{i\}, \\
w_{ij} &=& \frac{1-w_B}{n} \text{ for each } i \in B\text{ and } j\in N\setminus \{i\}.
\end{eqnarray*}
Note that this additional assumption does not affect any of the results of this proposition; we will investigate the role played by homophily in meetings in the next section.

Let us define $\eta:=n_A/(n_A+n_B)>0$ and $1-\eta :=n_B/n_A+n_B$ the share of the group $A$ and $B$, respectively. Then:
\begin{proposition}\label{prop:groupsize}
Let us consider the opinion dynamics in a society with two homogeneous groups and no homophily. Then, the opinion of the representative agents of each group at steady states are, respectively,
\begin{align}
     {\mu}_{A}=&
     \left( \frac{w_A -(1 -\eta)\big(w_A+w_B -2w_A w_B\big)}{1- \eta (1-w_A)- (1 -\eta )(1-w_B)} \right) \theta^* ;
     \label{eq:opinionssA}
     \\
     {\mu}_{B}=&
     \left( \frac{ w_B-\eta( w_A+w_B-2w_A w_B )}{1- \eta (1-w_A)- (1 -\eta )(1-w_B)} \right) \theta^*.
          \label{eq:opinionssB}
\end{align}
Moreover, an agent's \begin{revs}long-run\end{revs} opinion is:\\
\noindent \textit{(i)} increasing in the population share of her group, i.e., $\partial \mu_{A}/\partial \eta>0$ while $\partial \mu_{B}/\partial \eta<0$; \\
\noindent \textit{(ii)} increasing in the weight agents in her group give to the unbiased signal, i.e., $\partial \mu_{A}/\partial w_A>0$ and $\partial \mu_{B}/\partial w_B>0$; \\
\noindent \textit{(iii)} decreasing in the weight agents in the other group give to the unbiased signal, i.e., $\partial \mu_{A}/\partial w_B<0$ and $\partial \mu_{B}/\partial w_A<0$.
\end{proposition}

\begin{proof}
Given the assumptions, for a generic $i \in A$
\begin{align*}
    \sum_{j\in A} w_{ij}\mu_{j,t-1}= &\eta (1-w_A) \frac{ \sum_{j\in A } \mu_{j,t-1}}{n_A},
    \\
     \sum_{z\in B} w_{iz}\mu_{z,t-1}= & (1-\eta)(1-w_A) \frac{ \sum_{z\in B} \mu_{z,t-1}}{n_B}.
\end{align*}
Then, applying the same reasoning for agents in $B$ equations \eqref{eq:updatingA}
 and \eqref{eq:updatingB} become
  \begin{align*}
     \bm{\mu}_{A,t}=& \eta (1-w_A) \bar{\bm{\mu}}_{A,t-1} - (1-\eta)(1-w_A) \bar{\bm{\mu}}_{B,t-1} + w_A \bm{\theta}^* ;
     \\
     \bm{\mu}_{B,t}=& (1-\eta)(1-w_B) \bar{\bm{\mu}}_{B,t-1} - \eta(1-w_B) \bar{\bm{\mu}}_{A,t-1} + w_B \bm{\theta}^* .
 \end{align*}
The new opinion of each generic agent $i \in C $ at each time $t$, $\mu_{C,t}$, is a linear combination of the true state of the world $\theta^*$, the average opinion of members of own group $\bar{\mu}_{C,t-1}$, and the average opinion of members of the other group $\bar{\mu}_{\mathcal{C} \setminus C,t-1}$. 
\\
Given the assumptions of regularity of the network and the in-group homogeneity of $w_i$ (i.e., for all $i \in C$ and $C \in \mathcal{C}$, $w_i=w_C$). Therefore, for any $t\geq 1$ the opinions of agents in the same group are homogeneous, namely, for all $i \in C$ and $C \in \mathcal{C}$, $\mu_{i,t} =\mu_{C,t}$.
\\
Since the two groups are homogeneous, we can study the opinion dynamics for the representative agents of the two groups---i.e., we study the dynamics of $\bm{\mu}_{t}=({\mu}_{A,t}, {\mu}_{B,t})$.
\begin{equation*}
 \Rightarrow \bm{\mu}_{t}= \begin{bmatrix}
 \eta (1-w_A) & -(1 -\eta )(1-w_A) \\
- \eta (1-w_B) & (1 -\eta )(1-w_B)
\end{bmatrix} \bar{\bm{\mu}}_{t-1} +\begin{bmatrix}
w_A \\
w_B
\end{bmatrix} \theta^* 
\end{equation*}
\begin{revs}In the long-run,\end{revs} $\bm{\mu}_{t}=\bm{\mu}_{t-1}=\bar{\bm{\mu}}_{t}$
\begin{align}
\bm{\mu}= & \begin{bmatrix}
1- \eta (1-w_A) & (1 -\eta )(1-w_A) \\
 \eta (1-w_B) & 1- (1 -\eta )(1-w_B)
\end{bmatrix}^{-1} \begin{bmatrix}
w_A \\
w_B
\end{bmatrix} \theta^* 
 \notag \\
 \Rightarrow 
 & \frac{1}{Det} \begin{bmatrix}
 1- (1 -\eta )(1-w_B) & -(1 -\eta )(1-w_A) \\
 -\eta (1-w_B) &1- \eta (1-w_A) 
\end{bmatrix} \begin{bmatrix}
w_A \\
w_B
\end{bmatrix} \theta^*
  \label{eq:opinionss}
\end{align}
\begin{align*}
  Det:=&\big(1- \eta (1-w_A)\big)\big(1- (1 -\eta )(1-w_B)\big)-\big((1 -\eta )(1-w_A)\big)\big( \eta (1-w_B)\big)
  \\
  = &\big(1- \eta (1-w_A)- (1 -\eta )(1-w_B) + \eta(1 -\eta ) (1-w_A) (1-w_B)\big)\\ &-\eta (1 -\eta )(1-w_A) (1-w_B)
   \\
  = &1- \eta (1-w_A)- (1 -\eta )(1-w_B) 
\end{align*}
Substituting the determinant in \eqref{eq:opinionss}, we get \eqref{eq:opinionssA} and \eqref{eq:opinionssB}. 
\\
Let us compute the derivatives of \begin{revs}the long-run\end{revs} opinion for the representative agent of group $A$:
\begin{align*}
\frac{\partial \mu_{A}}{\partial \eta}= & \frac{2(1-w_A)w_A w_B}{(w_A \eta + w_B(1-\eta))^2}>0;
\\
\frac{\partial \mu_{A}}{\partial w_A} = & \frac{2 w_B(w_B(1-\eta) + \eta )(1-\eta)}{(w_A \eta + w_B(1-\eta))^2}>0;
\\
\frac{\partial \mu_{A}}{\partial w_B}= & -\frac{2 (1-w_A) w_A (1-\eta) \eta}{(w_A \eta + w_B(1-\eta))^2}<0.
\end{align*}
By symmetry we get the comparative statics for the representative agent of group $B$.
\end{proof}

Proposition \ref{prop:groupsize} reveals that increasing the size of a group or the weight that agents in a group give to the truth increases how close the opinions of agents in that group are to the truth. The role played by group share is related to the effect of inter-group conflict. Indeed, the larger the share of one's group, the lower is this effect, as there are fewer individuals she wants to disagree with. Hence, the closer she will get to the truth.

Furthermore, the more one cares about the unbiased source, the better the information she holds, and, quite intuitively, the closer her opinion should be to the truth. However, out-group conflict implies an opposite effect of the weight the members of the other group assign to the unbiased source of information. Indeed, as this increases, the less people of the other group care about the positions of others, which then becomes less central. As a result, in steady state, while out-group conflict effect pushes people's positions away from those of the other group, as the latter get more accurate, the worse one's opinion result to be.

As an example, we modify the population share of the homogeneous society with $4$ agents we studied in Example 1, by equalizing the size of the two groups. The resulting network, depicted in Figure \ref{fig:complete2}, contains two nodes for each group, who equally interact with all others. As shown in Table \ref{tab:complete2}, now all agents have the same Bonacich centrality in $\bm{W}$, so they all equally learn the truth (that is, $1$), converging to an opinion of $1/5$. Comparing this result with those in Table \ref{fig:complete2}, we see that the group that used to be the minority now is closer to the truth, but the opposite holds for the group that used to be the majority. More specifically, in this example, the new position is the middle ground ($1/5$) between the opinions previously held by the two groups ($-1/5$ and $3/5$, respectively). Hence, we see that collectively people are on average farther away from the truth in the more balanced society.

\begin{figure}[!ht]
   \centering
    \includegraphics[scale=0.3]{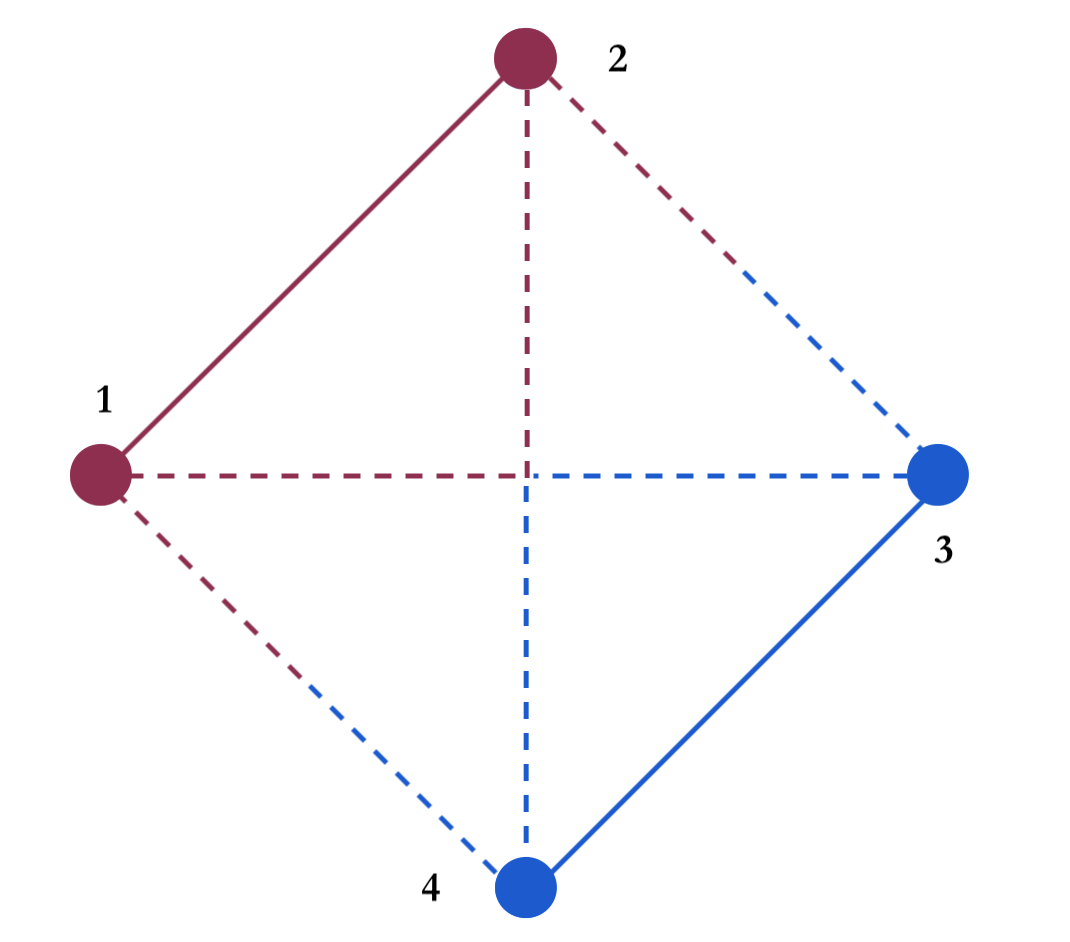}
    \caption{A signed complete network of $4$ agents. Dotted edges represent negative links, full edges represent positive links. Nodes in red belong to group $A$, those in blue to group $B$.}
  \label{fig:complete2}
\end{figure}
\begin{table}[!ht]
\begin{center}
\begin{tabular}{|c|c|c|}
\hline \rule{0pt}{10pt}
\emph{Nodes} & Bonacich centrality, $\mathbf{b}$ & \begin{revs}Long-run\end{revs} opinions, $\bm{\mu}$ \\ \hline \hline
1 & $5$ & $ 1/5$ 
\\ \hline
2 & $5$ & $1/5$ 
\\ \hline
3 & $5$ & $1/5$ 
\\\hline
4 & $5$ & $1/5$ 
\\ \hline
\end{tabular}%
\end{center}
\bigskip
\caption{Bonacich centralities and long-run opinions in the complete network depicted in Figure \ref{fig:complete2} with $\theta^*=1$, $\alpha=-\beta=1$ and $w_C=1/5$ for $C=A,B$.}
\label{tab:complete2}
\end{table}

In the following example, we look at the effect of an increase of the weight given to the unbiased source of information by the minority group $A$, $w_A$, again starting from Example 1. In Table \ref{tab:complete3}, we report the effect of increasing $w_A$ from $1/5$ to $1/2$. We can see that this increases the Bonacich centrality of agent $1$, the only member of group $A$, while it decreases those of agents in group $B$. Intuitively, as agent $1$ cares more about the unbiased source of information, she cares less about the opinion of agents in group $B$. Hence, as predicted by Proposition \ref{prop:groupsize}, group $A$ gets closer to the truth ($1$), while group $B$ gets further away from it.

\begin{table}[!ht]
\begin{center}
\begin{tabular}{|c|c|c|}
\hline \rule{0pt}{10pt}
\emph{Nodes} & Bonacich centrality, $\mathbf{b}$ & \begin{revs}Long-run\end{revs} opinions, $\bm{\mu}$\\ \hline \hline
1 & $31/11$ & $ 5/11$ 
\\ \hline
2 & $43/11$ & $3/11$ 
\\ \hline
3 & $43/11$ & $3/11$ 
\\\hline
4 & $43/11$ & $3/11$ 
\\ \hline
\end{tabular}%
\end{center}
\bigskip
\caption{Bonacich centralities and long-run opinions in the complete network depicted in Figure \ref{fig:complete} with $\theta^*=1$, $\alpha=-\beta=1$, $w_A=1/2$ and $w_B=1/5$.}
\label{tab:complete3}
\end{table}

\subsection{Homophily}

We now turn to the study of homophily, which we parametrize in the following way. Let us first define the probability with which individuals in group $A$ and $B$ interact among themselves as $\rho_A$ and $\rho_B$, respectively. Then, we follow \cite{coleman1958} and \cite{currarini2009} in defining \textit{inbreeding homophily} in the two groups as
\begin{eqnarray}\label{IH}
h_A=\frac{\rho_A-\eta}{1-\eta} \text{ and } h_B=\frac{\rho_B-(1-\eta)}{\eta}.
\end{eqnarray}
In words, inbreeding homophily measures the deviation of the probability of meeting someone of the same group with respect to the population average with respect to the maximal deviation that is possible, i.e., when agents only meet others of the same group.

Using \eqref{IH}, we can then rewrite the probability with which individuals in group $A$ and $B$ interact among themselves as
\begin{eqnarray*}\label{rhos}
\rho_A=h_A+(1-h_A)\eta \text{ and } \rho_B=h_B+(1-h_B)(1-\eta).
\end{eqnarray*}
Intuitively, when $h_A=0$, then the probability with which individuals in group $A$ interact among themselves is proportional to their share in the population. As homophily increases, so does this probability.

We can now state the following proposition.
\begin{proposition}\label{prop:homophily}
Let us consider the opinion dynamics in a society with two homogeneous groups. Then, the opinion of the representative agents at steady states satisfy the same comparative static with respect to group size $\eta$ and the weight of the public signals $w_A$ and $w_B$ as in Proposition \ref{prop:groupsize}. Moreover, an agent's \begin{revs}long-run\end{revs} opinion is:\\
\noindent \textit{(i)} increasing in the inbreeding homophily of her own group, i.e., $\partial \mu_{A}/\partial h_A>0$ and $\partial \mu_{B}/\partial h_B>0$;\\
\noindent \textit{(ii)} decreasing in the inbreeding homophily of the other group, i.e., $\partial \mu_{A}/\partial h_B<0$ and $\partial \mu_{B}/\partial h_A<0$.
\end{proposition}

\begin{proof}
Following the same argument of the proof of Proposition \ref{prop:groupsize}, in this case equations \eqref{eq:updatingA} and \eqref{eq:updatingB} become
\begin{align*}
    \bm{\mu}_{A,t}=& \rho_A (1-w_A) \bar{\bm{\mu}}_{A,t-1} - (1-\rho_A)(1-w_A) \bar{\bm{\mu}}_{B,t-1} + w_A \theta^* ;
     \\
    \bm{\mu}_{B,t}=& \rho_B(1-w_B) \bar{\bm{\mu}}_{B,t-1} - (1-\rho_B)(1-w_B) \bar{\bm{\mu}}_{A,t-1} + w_B \theta^*.
\end{align*}
\begin{equation*}
     \Rightarrow \bm{\mu}_{t}= \begin{bmatrix}
  \rho_A (1-w_A) & -(1 - \rho_A )(1-w_A) \\
- (1- \rho_B) (1-w_B) & \rho_B(1-w_B)
\end{bmatrix} \bar{\bm{\mu}}_{t-1} +\begin{bmatrix}
w_A \\
w_B
\end{bmatrix} \theta^* 
\end{equation*}
For the same argument of Proposition \ref{prop:groupsize} two groups are homogeneous, thus, we can study the opinion dynamics for the representative agents of the two groups---i.e., we study the dynamics of $\bm{\mu}_{t}=({\mu}_{A,t}, {\mu}_{B,t})$.
\\
At steady state
\begin{align*}
\bm{\mu}= & \begin{bmatrix}
 1-\rho_A (1-w_A) & (1 - \rho_A )(1-w_A) \\
 (1- \rho_B) (1-w_B) & 1-\rho_B(1-w_B)
\end{bmatrix}^{-1} \bar{\bm{\mu}}_{t-1} \begin{bmatrix}
w_A \\
w_B
\end{bmatrix} \theta^* 
 \\
 \Rightarrow 
 & \frac{1}{Det}
 \begin{bmatrix}
 1-\rho_B(1-w_B) &- (1 - \rho_A )(1-w_A) \\
 -(1- \rho_B) (1-w_B) &1-\rho_A (1-w_A) 
\end{bmatrix} \begin{bmatrix}
w_A \\
w_B
\end{bmatrix} \theta^*
 \end{align*}
\begin{align*}
   Det=& \big(1-\rho_A (1-w_A)\big)\big( 1-\rho_B(1-w_B)\big)- (1 - \rho_A )(1-w_A)(1- \rho_B) (1-w_B) 
    \\
    =& w_A(1-\rho_B)+w_B(1-\rho_A)-w_A w_B (1-\rho_A-\rho_B)
\end{align*}
Thus, the opinion of the average agents of the two groups \begin{revs}in the long-run\end{revs} are
\begin{align*}
     {\mu}_{A}=&
     \left( \frac{(1-h_B)w_A \eta- w_B(1-h_A(1-w_A)(1-\eta)-\eta-w_A(2-(2-h_B)\eta))}{(1-h_B)w_A \eta+w_B(1-h_A(1-w_A)(1-\eta)-\eta+h_B w_B \eta} \right) \theta^*,
     \\
     {\mu}_{B}=&
     \left( \frac{-(1-h_B)w_A \eta+ w_B(1-h_A(1-w_A)(1-\eta)-\eta+w_A\eta(2-h_B)}{(1-h_B)w_A \eta+w_B(1-h_A(1-w_A)(1-\eta)-\eta+h_B w_B \eta} \right) \theta^*.
\end{align*}
Before to proceed with comparative statics let us define $\delta:=(1-h_B)w_A \eta+w_B(1-h_A(1-w_A)(1-\eta)-\eta+h_B w_B \eta$.
\\
Let us compute the derivatives of \begin{revs}the long-run\end{revs} opinion for the representative agent of group $A$:
 \begin{align*}
     \frac{\partial \mu_{A}}{\partial \eta}=&\frac{2(1-h_A)(1-w_A)w_A(1-h_B(1-w_B))w_B}{\delta^2}>0,
     \\
     \frac{\partial \mu_{A}}{\partial h_A}=&\frac{2(1-w_A)w_A w_B (1-\eta)(w_B+(1-h_B)(1-w_B)\eta)}{\delta^2}>0,
     \\
     \frac{\partial \mu_{A}}{\partial h_B}=& -\frac{2(1-h_A)(1-w_A)(1-w_B) w_A w_B (1-\eta)\eta}{\delta^2}<0, 
     \\
     \frac{\partial \mu_{A}}{\partial w_A}=& \frac{2(1-h_A)w_B(1-\eta)(w_B+(1-h_B)(1-w_B)\eta)}{\delta^2}>0, 
     \\
     \frac{\partial \mu_{A}}{\partial w_B}=& -\frac{2(1-h_A)(1-h_B)(1-w_A)w_A(1-\eta)\eta}{\delta^2}<0,
 \end{align*}
By symmetry, we get the comparative statics for the representative agent of group $B$.
\end{proof}

The results of Proposition \ref{prop:homophily} on the effect of homophily emerge because of how a variation in the frequency of interaction within and across groups modulates the identity vs. the out-group conflict effects. As people interact more with people of their own group, the desire of within-group coordination brought about by the identity effect becomes more salient; as a result, people in the same group gets closer to the truth. On the contrary, when people interact more with others with whom they want to disagree, as postulated by the out-group conflict effect, this brings them farther way from the truth, as everyone shares the same desire to guess the correct state of the world.

Hence, Proposition \ref{prop:homophily} stresses that less homophily is not necessarily associated to more precise information.

This is shown in the example reported in Table \ref{tab:complete4}, in which we see the effect of an increase in homophily of agents in group $B$ of Example 1 (see Figure \ref{fig:complete}). The centralities and \begin{revs}long-run\end{revs} opinions reported in Table \ref{tab:complete} refer to a society with no homophily, i.e., $h_A=h_B=0$. In Table \ref{tab:complete4}, we report the results of an increase of inbreeding homophily of group $B$, $h_B$, to $1/2$. As agents in group $B$ interact more among themselves, inter-group conflict becomes less salient. As a result, they get closer to the truth. On the contrary, agent $1$ of group $A$ is less important, so that $1$ gets farther away from the truth.

\begin{table}[!ht]
\begin{center}
\begin{tabular}{|c|c|c|c|}
\hline \rule{0pt}{10pt}
\emph{Nodes} & Bonacich centrality, $\mathbf{b}$ & \begin{revs}Long-run\end{revs} opinions, $\bm{\mu}$ \\ \hline \hline
1 & $5$ & $ -1/3$ 
\\ \hline
2 & $5$ & $7/9$ 
\\ \hline
3 & $5$ & $7/9$ 
\\\hline
4 & $5$ & $7/9$ 
\\ \hline
\end{tabular}%
\end{center}
\bigskip
\caption{Bonacich centralities and long-run opinions in the complete network depicted in Figure \ref{fig:complete} with $\theta^*=1$, $\alpha=-\beta=1$, $w_C=1/5$ for $C=A,B$ and $h_B=1/2$.}
\label{tab:complete4}
\end{table}

\begin{revs}

\section{Discussion}\label{discussion} 

First, note that, even if we focus on a society with only two groups to keep the notation as simple as possible, the results established in Proposition \ref{prop:1} extend to weakly structurally balanced networks, i.e., to the existence of any number of opposed groups of mutual friends. More in detail, in the same way as we can think about a society with only positive links as one with only one group, we can represent any pattern of positive and negative influence by properly defining group membership.

Additionally, the results established in Proposition \ref{prop:1} extend to heterogeneous intensities of the in-group identity, $\alpha$, and out-group conflict, $\beta$. Indeed, in Section \ref{model} we define the matrix of the opinion exchange network $\tilde{\bm{W}}$ and the weights that agents give to the true state of the world, $\bm{\tilde{w}}$, using homogeneous values for these parameters for two reasons: to simplify the notation and to make their effects more transparent. However, we do not use this assumption in the derivation of Proposition \ref{prop:1}. Hence, its results extend to a more general definitions of $\tilde{\bm{W}}$ and $\bm{\tilde{w}}$.

Finally, in the benchmark model presented in Section \ref{model}, we have assumed that all individuals have access to an unbiased source of information. This serves as an anchor for their beliefs. It is however reasonable to think that some people have rather access to biased sources of information.

To address this concern, we can modify the model to include a bias $\xi \in \mathbb{R}$ in the source of information of individuals of one of the two groups, without loss of generality, group $B$.

Thus, equation \eqref{eq:dyn} can be written as
\begin{equation}\label{eq:dyn2}
    \begin{bmatrix}
    \bm{\mu}_{A,t}     \\       \bm{\mu}_{B,t}    \end{bmatrix}= \bm{\tilde{W}}     \begin{bmatrix}
    \bm{\mu}_{A,t-1}     \\       \bm{\mu}_{B,t-1}    \end{bmatrix} + \begin{bmatrix}
 \tilde{\mathbf{w}}_A \theta^*    \\        \tilde{\mathbf{w}}_B (\theta^* +b)   \end{bmatrix} .
\end{equation}
\\
\\
Defining $ \tilde{b}_i:=\sum_{j \in N}\tilde{b}_{ij} w_j$ and $ \tilde{b}^B_i:=\sum_{k \in B}\tilde{b}_{ij} w_k$, the next proposition characterizes the log-run opinion vector in such a case.

\begin{proposition}
\label{prop:bias}
The opinions dynamics defined in equation \eqref{eq:dyn2} always converge and, at the steady state, we have 
\begin{equation}\label{eq:ss2}
    \bm{\mu}= \left(\bm{I}-\bm{\tilde{W}} \right)^{-1}   \begin{bmatrix}
 \tilde{\mathbf{w}}_A \theta^*    \\        \tilde{\mathbf{w}}_B (\theta^* +\xi)   \end{bmatrix}.
\end{equation}
\\
\\
For the generic agent $i \in N$ we can write
\begin{equation}
    \mu_i=\tilde{b}_i \theta^* + \tilde{b}^B_{i} \xi.
\label{eq:ss2i}
\end{equation}
\end{proposition}
\begin{proof}
The Proof of this proposition trivially follows from the argument of the Proof of Proposition \ref{prop:1}.
\end{proof}

From Proposition \ref{prop:bias} we can see how the lung-run opinions depends also on the bias $\xi$ and how it propagates in the whole network through the directly biased agents, belonging to group $B$. Indeed, the long-run opinion of generic agent $i$ is the linear combination of the truth $\theta^*$ and the bias $\xi$, where the weights are $i$'s $weighted$ Bonacich centrality with respect to the whole society $N$ and only agents belonging to $B$, respectively.

Therefore, when as in the example below, $\tilde{b}_{ij}$ is  positive for each $i \in B$ and negative  for each $i \in A$ the opinions of agents belonging to the two groups are split and distributed around $\theta^*$. Interestingly, agents with the exogenous bias may be closer to the truth that the unbiased one. 

\begin{table}[!ht]
\begin{center}
\begin{tabular}{|c|c|c|}
\hline \rule{0pt}{10pt}
\emph{Nodes} & Bonacich centrality, $\mathbf{b}$ & \begin{revs}Long-run\end{revs} opinions, $\bm{\mu}$ \\ \hline \hline
1 & $5$ & $ -0.4$ 
\\ \hline
2 & $5$ & $-0.4$ 
\\ \hline
3 & $5$ & $1.1$ 
\\\hline
4 & $5$ & $1.1$ 
\\ \hline
\end{tabular}%
\end{center}
\caption{Bonacich centralities and long-run opinions in the complete network depicted in Figure \ref{fig:complete2} with $\theta^*=1$, $\xi=1.5$, $\alpha=-\beta=1$ and $w_C=1/5$ for $C=A,B$.}
\end{table}

\end{revs}

\section{Conclusions} \label{conclusions}

In this paper, we have developed a model of opinion dynamics with social learning in which agents wish to have similar views to those the same group, but to have different views from those of another group. Starting from an interaction matrix, we derive a matrix of opinion exchange, which is a structurally balanced signed network.

We show that the opinion dynamics on this network has a unique steady state, in which agents' learning ability depends on their Bonacich centrality in the opinion exchange network. Opinion leaders in social learning with out-group conflict differ from those predicted when there is no conflict. Furthermore, the farther away is an agent from the other group, the closer her opinion to the truth is.

We show that, in societies with symmetric interaction structures, the out-group conflict effect generates novel predictions on the effect on opinion dynamics of group size, the weight people give to the unbiased source of information and homophily.

This paper represents a first step in the understanding of opinion dynamics on signed networks. Future research should study the convergence properties of opinion dynamics when the true state of the world is uncertain. Furthermore, in our model the network of socialization is given. It would be important to understand what are the incentives to interact with people of the same vs. the other group in the presence of in-group identity and out-group conflict.

\bigskip 

\noindent \textbf{Acknowledgments.} We would like to thank for financially supporting this project: Research Foundation - Flanders (FWO) through grants 1258321N and G029621N; the Australian Research Council (ARC) through grant DP200102547.


\end{document}